\documentclass{article}

\usepackage{amssymb}
\usepackage{amsmath}
\usepackage{amsfonts}
\usepackage{amsthm}
\usepackage{a4wide}
\usepackage{hyperref}
\usepackage{tikz}
\usetikzlibrary{trees,arrows,positioning}

\newtheorem{theorem}{Theorem}
\newtheorem{corollary}{Corollary}
\newtheorem{definition}{Definition}
\newtheorem{algorithm}{Algorithm}

\author{Andreas Rosowski\\ University of Siegen, Germany\\ andreas.rosowski@web.de}

\begin{document}

\title{On Fast Computation of a Circulant Matrix-Vector Product}

\maketitle
\begin{abstract}
This paper deals with circulant matrices. It is shown that a circulant matrix can be multiplied by a vector in time $O(n \log(n))$ in a ring with roots of unity without making use of an FFT algorithm. With our algorithm we achieve a speedup of a factor of about $2.25$ for the multiplication of two polynomials with integer coefficients compared to multiplication by an FFT algorithm. Moreover this paper discusses multiplication of large integers as further application.
\end{abstract}

\section{Introduction}
\subsection{Related Work}
In this paper we study circulant matrices and some of their applications. Circulant matrices are well studied and for instance it was shown that the product of a circulant matrix and a vector of size $n$ can be computed using $O(n \log(n))$ operations \cite{golub_loan} using FFT. Golub and van Loan also showed that a Toeplitz matrix and a vector can be multiplied in time $O(n \log(n))$ by expressing a Toeplitz matrix as circulant matrix. Circulant matrices and related structured matrices were also studied by several authors \cite{attenduross, blockcirculant, realcirculant, pan, yelim}.\\

Beginning with the famous algorithm by Karatsuba \cite{karatsuba} which yields a complexity of $O(n^{\log_2(3)})$ for the multiplication of two $n$ bit numbers the complexity was drastically reduced. Compared to the naiv schoolbook multiplication which needs $O(n^2)$ operations the method of Karatsuba was a hugh improvement. A generalization of Karatsuba's method was done by Toom \cite{toom} which yields an algorithm with a computational complexity of $O(n^{1+\varepsilon})$ for $\varepsilon > 0$. This algorithm works as follows: We choose two numbers $k_1, k_2 \geq 2$. Then we consider the input as two polynomials of length $k_1$ respectively $k_2$. Then the product of those polynomials will have length $k_1 + k_2 - 1$. It is a well known fact that evaluation and interpolation of a polynomial with degree $n$ works with $n+1$ pairwise different points. Thus, it suffice to evaluate the polynomials at $k_1 + k_2 - 1$ pairwise different points. The component-wise multiplication of the points is then realized as recursive calls. Interpolation yields the final coefficients. For $k_1 = k_2$ the running time is $O(n^{\log_{k_1}(2k_1 - 1)})$. We do not want to go into further details but let us also mention that the fastest algorithm used in praxis is the algorithm by Sch\"onhage and Strassen \cite{schoenhage_strassen} which yields a running time of $O(n \log(n) \log \log (n))$.\\

The main goal of this paper is to avoid the usage of an FFT algorithm for computing the product of a circulant matrix and a vector as well as for the multiplication of two polynomials. As further application this paper discusses a modification of the Sch\"onhage and Strassen algorithm \cite{schoenhage_strassen}. To archieve this we present an algorithm that computes the product of a circulant matrix and a vector in time $O(n \log(n))$ without making use of FFT. So called $f$-circulant matrices play an important role in our algorithm.

\subsection{Preliminaries}
In this paper we will consider so called circulant matrices. We give a formal definition:
\begin{definition}
A circulant matrix is an $n \times n$ matrix. For every entry $a_{ij}$ with $1 \leq i,j \leq n$ the following property holds: $a_{ij} = a_{i'j'}$, in which we have $i' = (i \mod n)+1$ and $j'= (j \mod n)+1$.
\end{definition}

Obviously a circulant matrix can be represented by a vector of size $n$. Thus, operations on circulant matrices such as additions can be computed in linear time. We are interested in computing the product of a circulant matrix and a vector. The usual way of computing such a product is by using the FFT since \cite{golub_loan}
\begin{displaymath}
A = F_n^{-1}diag(F_na)F_n
\end{displaymath}
in which $A$ denotes a circulant matrix, $a$ denotes the vector representing $A$ and $F_n$ denotes the DFT-matrix. Thus, the product of a circulant matrix and a vector can be computed using three FFTs in time $O(n \log(n))$. See \cite{FFT} for an FFT algorithm. The main goal of this paper is to avoid the use of three FFTs. Instead we give an algorithm computing the product of a circulant matrix and a vector directly. For this we need a special form of circulant matrices. The matrices we use are so called $f$-circulant matrices for a number $f$ and were defined in \cite{pan}.
\begin{definition}\label{Dzyklisch}
An $f$-circulant matrix is an $n \times n$ matrix. Furthermore for every entry $a_{ij}$ with $1 \leq i,j, \leq n$ the following two properties holds:
\begin{equation}\label{(1)}
a_{ij}= a_{i'j'} \text{~~~~for } i' = i+1 \leq n \text{ and } j' = j+1 \leq n
\end{equation}
\begin{equation}\label{(2)}
a_{ij}f= a_{i'j'} \text{~~~~for } j = n \text{ and } i' = i+1 \leq n \text{ and } j' = 1
\end{equation}
\end{definition}

Note that a circulant $1$-matrix is a usual circulant matrix. We give an example for an $f$-circulant matrix. Consider this $3 \times 3$ matrix:
\begin{align*}
\begin{bmatrix}
a_1 & a_2 & a_3\\
a_3f & a_1 & a_2\\
a_2f & a_3f & a_1
\end{bmatrix}
\end{align*}

Throughout the rest of the paper we identify circulant matrices and $f$-circulant matrices by row vectors and not by column vectors as it is usually be done.

\section{Fast Multiplication of Circulant Matrices}
\begin{theorem}
Let $R$ be a ring. Let $f \in R$ and let $n$ be a power of two. Then we have: If there is an $n$-th root of $f$ in $R$, an $n$-th root of unity in $R$ and there are $2^{-1}$ and $f^{-1}$ then the product of an $f$-circulant $n \times n$ matrix and a vector can be computed with $O(n \log(n))$ operations.
\end{theorem}
\begin{proof}
Let
\begin{displaymath}
A = \begin{bmatrix}
A_1 & A_2\\
A_2f & A_1
\end{bmatrix}
\end{displaymath}
and let $b = (b_1, b_2)^T$ and let $Ab = (c_1, c_2)^T$. We consider the algorithm for computing the product $Ab$:
\begin{algorithm}\label{Dalgorithm}
Input: $f$-circulant matrix $A$, vector $b$
\begin{align*}
M_1&:=(A_1 + A_2\sqrt{f})(b_1\sqrt{f} + b_2)\\
M_2&:=(A_1 - A_2\sqrt{f})(b_1\sqrt{f} - b_2)
\end{align*}
Output:
\begin{align*}
c_1&=\frac{M_1 + M_2}{2\sqrt{f}}\\
c_2&=\frac{M_1 + M_2}{2} - M_2
\end{align*}
\end{algorithm}

Now, let us show that recursive applications preserve the structure of the matrix. Let $n \geq 4$ be a power of two, $A$ be an $f$-circulant $n \times n$ matrix and $b$ be an $n \times 1$ vector. We apply algorithm \ref{Dalgorithm} to compute the product of $A$ and $b$. With $a_{ij}$ we denote the entries of $A$. Let $A_{M_1} = A_1 + A_2\sqrt{f}$ and let $A_{M_2} = A_1 - A_2\sqrt{f}$. For both matrices \eqref{(1)} of definition \ref{Dzyklisch} clearly holds . Without loss of generality let $i=1$. Let $j \in \{2,\dots,n/2\}$. Since the diagonal is irrelevant for preserving the structure we do not need to consider $j=1$. First we consider the matrix $A_{M_1}$. This matrix consists of the addition $a_{1j} + a_{1(j+n/2)}\sqrt{f}$. Because of \eqref{(1)} of definition \ref{Dzyklisch} we have:
\begin{align*}
a_{1j} + a_{1(j+n/2)}\sqrt{f}&= a_{2(j+1)} + a_{2(j+1+n/2)}\sqrt{f}\\
&= a_{3(j+2)} + a_{3(j+2+n/2)}\sqrt{f}\\
&\vdots\\
&= a_{(1-j+n/2)(n/2)} + a_{(1-j+n/2)n}\sqrt{f}\\
&= a_{(2-j+n/2)(1+n/2)} + a_{(1-j+n/2)n}\sqrt{f}
\end{align*}
We multiply both sides with $\sqrt{f}$ and obtain:
\begin{displaymath}
a_{1j}\sqrt{f} + a_{1(j+n/2)}f = a_{(2-j+n/2)(1+n/2)}\sqrt{f} + a_{(1-j+n/2)n}f
\end{displaymath}
By \eqref{(2)} of definition \ref{Dzyklisch} we get $a_{(1-j+n/2)n}f = a_{(2-j+n/2)1}$ and it follows:
\begin{displaymath}
a_{1j}\sqrt{f} + a_{1(j+n/2)}f = a_{(2-j+n/2)(1+n/2)}\sqrt{f} + a_{(2-j+n/2)1}
\end{displaymath}
Furthermore by \eqref{(1)} of definition \ref{Dzyklisch} we have $a_{1j}\sqrt{f} + a_{1(j+n/2)}f = a_{(1-j+n/2)(n/2)}\sqrt{f} + a_{(1-j+n/2)n}f$. By that we obtain:
\begin{displaymath}
a_{(1-j+n/2)(n/2)}\sqrt{f} + a_{(1-j+n/2)n}f = a_{(2-j+n/2)(1+n/2)}\sqrt{f} + a_{(2-j+n/2)1}
\end{displaymath}
Let $f' = \sqrt{f}$. Finally we get:
\begin{displaymath}
(a_{(1-j+n/2)(n/2)} + a_{(1-j+n/2)n}\sqrt{f})f' = a_{(2-j+n/2)(1+n/2)}\sqrt{f} + a_{(2-j+n/2)1}
\end{displaymath}
Hence, \eqref{(2)} of definition \ref{Dzyklisch} is preserved by the matrix $A_{M_1}$. Moreover $A_{M_1}$ can be seen as $f'$-circulant matrix and we are able to apply algorithm \ref{Dalgorithm} recursively. Let $f'' = -\sqrt{f}$. Analogously we can show that $A_{M_2}$ is an $f''$-circulant matrix.
\end{proof}
\begin{corollary}
Let $A$ be a circulant $n \times n$ matrix and let $b$ be an $n \times 1$ vector. Then the product $Ab$ can be computed using $O(n \log(n))$ operations without making use of FFT.
\end{corollary}\label{Dcorollary}
\begin{proof}
We consider $A$ as $1$-circulant matrix and use algorithm \ref{Dalgorithm} to compute the product $Ab$. If $n$ is a power of two all is clear. So, let us suppose $n$ to be not a power of two. Let $d$ be the smallest integer such that $2^d > n$ and let $N = 2^{d+1}$. We identify $A$ by the row vector $a = (a_1,\dots,a_n)$. Let $a' = (a_1,\dots,a_n,0,\dots,0,a_2,\dots,a_n)$ and $b'=(b_1,\dots,b_n,0,\dots,0)^T$ be vectors of length $N$. Let $A'$ be the circulant matrix identified by $a'$. Then the product $Ab$ clearly is computed by $A'b'$ if $b=(b_1,\dots,b_n)^T$.
\end{proof}

\section{Some Applications}
\subsection{Multiplication of Polynomials}\label{MultiplicationofPolynomials}
We implemented Algorithm \ref{Dalgorithm} to multiply two polynomials with integer coefficients. Let $a(x)$ and $b(x)$ be two polynomials with degree $d_a$ respectively $d_b$. Let $N = (d_a + 1)(d_b + 1) - 1$. If $N$ is not a power of two we apply Corollary \ref{Dcorollary} to fill up the input. So let us suppose $N$ to be a power of two. The multiplication of two polynomials forms a circulant matrix-vector product in a natural way. In the following we give a description of our settings. The implementation was done in java. We used $64$ bit to deal with multiplications of $32$ bit numbers. To avoid the use of complex numbers we decided to compute modulo a Mersenne prime. There is a suitable Mersenne prime $p:=2^{31}-1 = 2147483647$ with $31$ bit. For our algorithm the input size needs to be a power of two and also the $N$-th root of unity needs to be a power of two. Since $p+1$ is a power of two we can compute in $\mathbb{Z}/p\mathbb{Z}[\sqrt{3}]$ where we can find the generator $2 + \sqrt{3}$ \cite{forster}. Hence, we can find with $(2 + \sqrt{3})^{\frac{p+1}{N}}$ our $N$-th root of unity. For comparison we also implemented the well known FFT algorithm by Cooley and Tukey \cite{FFT}. Note that the computation in $\mathbb{Z}/p\mathbb{Z}[\sqrt{3}]$ does not affect the comparison since we implemented both algorithms in this field. Also note that the roots needed by Algorithm \ref{Dalgorithm} can be precomputed in linear time. The inverse roots do not need to be computed separately since they are already computed. For instance let $\sqrt{-1} = i$ then we have $(\sqrt[4]{i})^{-1} = -i\sqrt{i}\sqrt[4]{i}$ in which both numbers are precomputed roots. Let us compare the running times of the algorithms\footnote{The implementation of both algorithms can be found at https://github.com/hans152/polynomials}. In our comparison both polynomials have the same size (number of coefficients) which we call $n$. Since we have chosen $n$ to be a power of two $2n-1$ is not a power of two. Thus, we need to fill up the input. The procedure of filling up is the same in both implementations. By classic we mean the implementation using three FFTs and by circulant we mean Algorithm \ref{Dalgorithm}.
\begin{displaymath}
\begin{tabular}{|c|c|c|c|}
\hline
$n$ & classic & circulant & ratio\\
\hline
8 & 270ms & 124ms & 2.18\\
\hline
16 & 593ms & 265ms & 2.24\\
\hline
32 & 1295ms & 577ms & 2.24\\
\hline
64 & 2886ms & 1279ms & 2.26\\
\hline
128 & 6427ms & 2870ms & 2.24\\
\hline
256 & 15772ms & 6490ms & 2.43\\
\hline
512 & 39187ms & 16473ms & 2.38\\
\hline
\end{tabular}
\end{displaymath}
For every $n$ we multiplied $10000$ polynomials and ran both algorithms several times and always noted the smallest measured time. Note that our implementation can easily be adapted to multiply also polynomials with say real coefficients. Of course complex numbers are needed in this case.

\subsection{Multiplication of Large Integers}
In this section we discuss the multiplication of large integers as further application of algorithm \ref{Dalgorithm}. For this we briefly give an overview of the Sch\"onhage-Strassen algorithm \cite{schoenhage_strassen} first. We will not go into deeper details of the algorithm.

Suppose we get as input two arrays of length $n_1$ respectively $n_2$ of integers and let $n = n_1 + n_2 - 1$. Now split the arrays in roughly $\sqrt{n}$ parts such that the number of parts is a power of two. We call the number of parts $N$. Then we can see the parts as integer coefficients of two polynomials of length $N$. Choose a suitable residue class ring such that the module is a power of two plus one. Then we can find a $2N$-th root of unity which is also a power of two. This allows to compute operations with roots of unity by shifting. Finally use FFT to evaluate these two polynomials und compute the component-wise multiplications recursively. Use a third FFT for interpolation.

By algorithm \ref{Dalgorithm} we can avoid the usage of three FFTs and compute the product of the two polynomials directly. In contrast to section \ref{MultiplicationofPolynomials} the roots and inverse roots of algorithm \ref{Dalgorithm} do not need to be computed since the roots of unity are powers of two. Thus, we only need to know the number of bits that are to shift.

\section{Conclusion}
We have developed an algorithm to compute the product of a circulant matrix and a vector in time $O(n \log(n))$ without using an FFT algorithm. Furthermore we discussed the multiplication of two polynomials and the multiplication of large numbers as application of our algorithm. We implemented the algorithm for the multiplication of two polynomials and obtained indeed a speed up of a constant factor. This fact makes it plausible that our algorithm can also speed up the Sch\"onhage-Strassen algorithm \cite{schoenhage_strassen}. This would be of high practical interest. For instance for the search of Mersenne primes. The largest known prime number currently has around $25$ million decimal digits and the computations becoming increasingly expensive.

\bibliographystyle{abbrv}

\end{document}